\documentclass[conference]{IEEEtran}
\usepackage{cite}
\usepackage{amsmath,amssymb,amsfonts}
\usepackage{graphicx}
\usepackage{textcomp}
\usepackage{xcolor}
\def\BibTeX{{\rm B\kern-.05em{\sc i\kern-.025em b}\kern-.08em
    T\kern-.1667em\lower.7ex\hbox{E}\kern-.125emX}}

\usepackage[appendix=inline]{apxproof}
\newtheoremrep{theorem}{Theorem}
\newtheoremrep{lemma}{Lemma}
\newtheoremrep{proposition}{Proposition}
\newtheoremrep{corollary}{Corollary}

\usepackage{environ,etoolbox}
\newcommand{\MyAppendixBuffer}{}%
\NewEnviron{storetoappendix}{\gappto{\MyAppendixBuffer}{\BODY}}%

\usepackage{mathtools}
\usepackage{braket}
\usepackage{dblfloatfix}
\usepackage{multirow}
\usepackage{makecell}
\usepackage{diagbox}
\usepackage{placeins}
\usepackage{subcaption}
\usepackage[table]{xcolor}
\usepackage{algorithm}
\usepackage{algorithmic}
\usepackage{changepage}
\usepackage{amsthm}
\usepackage{enumitem}
\usepackage{hyperref} 
\usepackage{pifont} 

\def \VersionWithComments {}

\ifdefined \VersionWithComments
\usepackage{marginnote}
\newcommand{\marginX}{\marginnote{\huge{\quad\quad\textbf{!}\quad\quad}}}
\newcommand{\cyr}[1]{\mbox{}{\color{green!50!black}\marginX{}\textbf{[Yean-Ru}: #1]}}
\newcommand{\lsw}[1]{\mbox{}{\color{orange}\marginX{}\textbf{[Shang-Wei}: #1]}}
\newcommand{\ylh}[1]{\mbox{}{\color{purple}\marginX{}\textbf{[Lei-Han}: #1]}}
\newcommand{\cyc}[1]{\mbox{}{\color{cyan}\marginX{}\textbf{[Yu-Chung}: #1]}}
\newcommand{\instructions}[1]{{\color{red}\marginX{}\textbf{[Instructions: ``#1'']}}}
\newcommand{\reviewer}[2]{\mbox{}{\color{red}\marginX{}\textbf{[Reviewer #1}: ``#2'']}}
\newcommand{\todo}[1]{\mbox{}{\color{blue}{\marginX{}\textbf{TODO}\ifx#1\\\else:\ \fi #1}}} 
\else
\newcommand{\instructions}[1]{}
\newcommand{\cyr}[1]{}
\newcommand{\lsw}[1]{}
\newcommand{\ylh}[1]{}
\newcommand{\cyc}[1]{}
\newcommand{\reviewer}[2]{}
\newcommand{\todo}[1]{}
\fi

\newcommand{\len}[1]{|#1|}
\newcommand{\MATHFN}[1]{\mathsf{#1}} 
\newcommand{\CKTFN}[1]{\mathtt{#1}}  
\newcommand{\GATE}[1]{\mathsf{#1}}   
\newcommand{\TGATE}{\mathtt{T}}   

\newcommand{\TCfoVtw}[0]{-40.00\%}
\newcommand{\TCfoVtr}[0]{15.15\%}
\newcommand{\TCTfoVtw}[0]{40.00\%}
\newcommand{\TCTfoVtr}[0]{63.64\%}
\newcommand{\TCTNfoVtw}[0]{40.00\%}
\newcommand{\TCTNfoVtr}[0]{63.64\%}

\newcommand{\TDfoVtw}[0]{-20.00\%}
\newcommand{\TDfoVtr}[0]{20.00\%}
\newcommand{\TDTfoVtw}[0]{60.00\%}
\newcommand{\TDTfoVtr}[0]{73.33\%}
\newcommand{\TDTNfoVtw}[0]{60.00\%}
\newcommand{\TDTNfoVtr}[0]{73.33\%}

\newcommand{\AfoVtw}[0]{50.00\%}
\newcommand{\AfoVtr}[0]{0.00\%}
\newcommand{\ATfoVtw}[0]{50.00\%}
\newcommand{\ATfoVtr}[0]{0.00\%}
\newcommand{\ATNfoVtw}[0]{-100.00\%}
\newcommand{\ATNfoVtr}[0]{-300.00\%}

\newcommand{\DfoVtw}[0]{31.15\%}
\newcommand{\DfoVtr}[0]{37.31\%}
\newcommand{\DTfoVtw}[0]{47.54\%}
\newcommand{\DTfoVtr}[0]{52.24\%}
\newcommand{\DTNfoVtw}[0]{22.95\%}
\newcommand{\DTNfoVtr}[0]{29.85\%}
\begin{document}

\title{A Modular and T-Gate Efficient Architecture for Quantum Leading-Zero/One Counter}

\author{
 
 \IEEEauthorblockN{Lei-Han Yao$^{1,2}$, Shang-Wei Lin$^1$}
 \IEEEauthorblockA{\textit{Singapore Institute of Technology$^1$} 
 \\Singapore\\
 shangwei.lin@singaporetech.edu.sg}
 \and
 \IEEEauthorblockN{Yu-Chung Chen$^{1,2}$, and Yean-Ru Chen$^2$}
 \IEEEauthorblockA{\textit{National Cheng Kung University$^2$} \\Tainan, Taiwan (R.O.C.)\\
 chenyr@mail.ncku.edu.tw}
}

\maketitle

\begin{abstract}
The Quantum Leading-Zero/One Counter (QLZOC) is a fundamental component in quantum arithmetic, playing a critical role in normalization, floating-point units, dynamic range scaling, and logarithmic approximations. 
Conventional designs primarily rely on direct Boolean-to-quantum mapping, which results in inefficient resource utilization such as irregular gate growth and width-dependent resource overhead.
In this work, we propose a scalable, modular, and resource efficient architecture for QLZOC by reformulating the counting process into a sequence of systematic conditional bit-flip operations. 
Moreover, our design achieves functional polymorphism so that the same design can be easily toggled between zero and one detection, while ensuring seamless scalability to any bit-width without manual re-tuning. 
We further introduce a Parallel QLZOC (PQLZOC) variant and a Fan-Out optimized (FO-PQLZOC) design. 
In this work, we evaluate resource efficiency based on the classic criteria about T gates, including the number of total T gates being used (T-count) and the number of sequential T gate layers (T-depth). 
By exploiting the properties of all-zero/one qubit blocks and a hierarchical merge strategy, the proposed FO-PQLZOC reduces the T-depth from $\mathcal{O}(m)$ to $\mathcal{O}(\log_2 m)$, where $m$ is the input size. Comparative analysis demonstrates that our optimized architecture achieves a $40.00\%$ reduction in T-count and a $60.00\%$ reduction in T-depth over state-of-the-art designs, providing a high-performance, T-gate efficient solution for general-purpose quantum arithmetic processors.

\end{abstract}

\vspace{1em}

\section{Introduction} \label{sec:Introduction}
Efficient quantum arithmetic is fundamental to the execution of complex algorithms in scientific computing, signal processing, high-performance data analysis, and constraint-solving tasks~\cite{APPforAinSC,QSMT,3D_FP}. As quantum processors scale, there is an increasing demand for dynamic range scaling and floating-point operations~\cite{Q_Fadder_LZC,Q_Fadder_LZD,Q_Fdivision_LZD,Q_Fdivision,3D_FP}, where the Quantum Leading-Zero/One Counter (QLZOC) serves as a core infrastructure. The task of QLZOC is to identify the position of the first non-zero (or non-one) bit to determine the number of bit-shifting operations required, which is indispensable for mantissa normalization, arithmetic alignment, and numerical compression.

Existing quantum LZC or LOC implementations typically rely on a direct Boolean-to-quantum mapping flow \cite{Q_LZD, Q_Fdivision_LZD, Q_Fadder_LZC, Q_Fadder_LZD}, which necessitates a complete architectural revamp whenever the input size changes. This is because the underlying Boolean expressions scale irregularly with the input register size $n$. For instance, the output $\gamma$ of a QLZD with $3$-qubit input $X = (X_2X_1X_0)$ are represented by:
\[
    \gamma_1 = \overline{X_2} \land \overline{X_1}, \quad \gamma_0 = \overline{X_2}\land (X_1 \lor  \overline{X_0}), 
\]
where $\overline{X_i}$ is the bitwise complementation. 
However, expanding it to a $4$-qubit input $X = (X_3X_2X_1X_0)$ requires a redefined logic to incorporate the additional most significant bit (MSB):
\[
\begin{gathered}
    \gamma_2 = \overline{X_3} \land \overline{X_2} \land \overline{X_1} \land \overline{X_0} , \quad 
    \gamma_1 = \overline{X_3} \land \overline{X_2} \land (X_1 \lor X_0) , \\
    \gamma_0 = \overline{X_3} \land (X_2 \lor (\overline{X_1} \land X_0))
\end{gathered}
\]
Under this paradigm, each Boolean $\GATE{AND}$/$\GATE{OR}$ operator generally requires at least one $\GATE{CCX}$ gate and an additional ancilla qubit to store intermediate results. If restoring the value of ancilla qubits is required to reclaim quantum resources, the usage of gates could be easily doubled (requiring at least another $\GATE{CCX}$ gate per operation). 
Although leveraging intermediate Boolean values could potentially reduce the number of ancilla qubits and circuit depth, it requires additional complex optimization efforts. 
Consequently, these approaches often struggle with resource efficiency issues. Notice that, in this work, we evaluate resource efficiency based on classical criteria about T gates, including the number of total
T gates being used (T-count) and the number of sequential T gate layers (T-depth). 
This is because T gate, compared to common Clifford gates, is exceptionally costly and difficult to implement in fault-tolerant and error-corrected systems. Thus, minimizing T-count and T-depth is essential in quantum circuit design to improve efficiency.


\begin{figure}[t]
    \centering
    \includegraphics[width=0.9\linewidth]{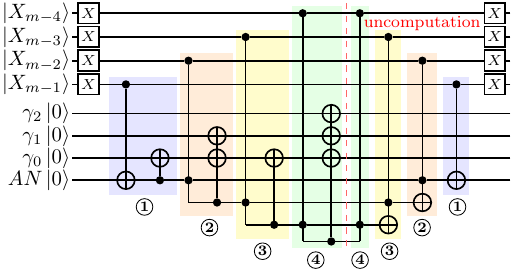}
    \caption{Modular Sequential QLZC for 1 to 4 input size.}
    \label{fig:eg_LZC}
\end{figure}
In this work, we propose a modular and resource-efficient architecture for QLZOC by leveraging the intrinsic conditional bit-flip logic of quantum controlled gates ($\GATE{CX}$, $\GATE{CCX}$, and $\GATE{MCX}$) and the bit-toggle patterns inherent in binary incrementing---the process of incrementing an integer value where multiple lower-order bits are flipped due to carry propagation.

Fig.~\ref{fig:eg_LZC} shows how our approach constructs QLZOCs incrementally for input qubit size from $1$ to $4$.
For a $1$-qubit QLZOC, only the components labeled by \ding{172} are required. To construct a $2$-qubit QLZOC, the components labeled by \ding{173} are added on top of the $1$-qubit QLZOC, i.e., the final quantum circuit of the $2$-qubit QLZOC is \ding{172} + \ding{173}. Similarly, the quantum circuit for a $4$-qubit QLZOC is \ding{172} + \ding{173} + \ding{174} + \ding{175}.
Each additional input qubit requires one $\GATE{CCX}$ and few $\GATE{CX}$ gates for the main computation, with an additional $\GATE{CCX}$ gate if restoring the values of ancilla qubits is necessary.
Our approach ensures that the circuit complexity scales linearly with the input size, enabling straightforward expansion. 
We summarize our technical contributions as follows:
\begin{itemize}
    \item \textbf{Scalable, Garbage-Free Modular Construction}: We propose a sequential architecture (TA-OP QLZOC) integrating \emph{Temporary Logical-AND} gates that scales linearly with bit-width $m$ ($\mathcal{O}(m)$ for $\TGATE$-count, depth, and width).   
    Unlike traditional designs, our approach has systematic structural patterns, which allows seamless scaling to any input size $m$ without altering the underlying design logic and produces zero garbage ancilla, making it highly suitable for qubit-constrained environments.
    
    \item \textbf{Hierarchical Parallelization and Fan-out Control}: We introduce a parallelized variant (PQLZOC) that exploits the property: an all-zero/one $m$-qubit block yields a predictable binary state. By substituting conventional controlled adders with $(\lg(m)\!+\!1)$ $\GATE{CCX}$ $+1$ $\GATE{CX}$ gates, 
    the $\TGATE$-depth and overall depth could be reduced to $\mathcal{O}((\log_2 m)^2)$. Integrating fan-out control (FO-PQLZOC) further reduces $\TGATE$-depth and overall depth to $\mathcal{O}(\log_2 m)$ and $\mathcal{O}(\log_2 m \log_2 \log_2 m)$, respectively.

    \item \textbf{Polymorphic Architectural Design}: By leveraging the inherent logical symmetry between zero and one detection, our design achieves functional polymorphism. This allows for a seamless transition between QLZC and QLOC configurations with minimal Clifford-gate adjustments, maximizing hardware utility.
    
    \item \textbf{Resource Efficiency}: Beyond theoretical complexity, our designs outperform existing state-of-the-art implementations across all physical metrics. For 4-qubit inputs, our optimized design outperforms~\cite{Q_LZD} by 40.00\%, 60.00\%, 47.54\%, and 20.00\% in $\TGATE$-count, $\TGATE$-depth, overall depth, and circuit width, respectively, proving its efficiency in resource for practical hardware realization.
\end{itemize}

The remainder of this paper is structured as follows: 
Sec.~\ref{sec:Preliminaries} reviews the mathematical notations and theoretical foundations. 
Sec.~\ref{sec:Methodology} details the proposed QLZOC methodology, including the quantum-oriented reformulation, hierarchical parallelization, and fan-out optimization strategies. 
Sec.~\ref{sec:RelatedWork} provides a comprehensive comparison between our designs and prior state-of-the-art work, with granular resource derivations deferred to the Supplementary~\cite{Supplementary}. 
Sec.~\ref{sec:ExperimentalResults} presents simulation data and comparative benchmarks validating our theoretical framework across various bit-widths. 
Finally, Sec.~\ref{sec:Conclusions} concludes this paper and discusses future research directions.



\section{Preliminaries} \label{sec:Preliminaries}
\subsection{Mathematical Notation}
\begin{itemize}
\item $\mathbb{N}_0 \!=\! \{0, 1, 2,\dots\}$ and $\mathbb{N}_1 \!=\! \{1,2,\ldots\} \!=\! \mathbb{N}_0\!\setminus\!\{0\}$ denote the sets of nonnegative and positive integers, respectively.

\item The expression $2^p$ denotes a power of two, where $p \in \mathbb{N}_0$.


\item $\len{S}$ denotes the bit/qubit length of a register $S$; e.g., for a 3-qubit register $\ket{X} = \ket{X_2 X_1 X_0}$, $\len{X} = 3$.

\item If--Then--Else function $\MATHFN{ITE}(c, x, y)$ returns $x$ if the Boolean condition $c=1$, and $y$ otherwise. Formally,
\[
\MATHFN{ITE}(c, x, y)
\;\triangleq\;
\begin{cases}
x, & c = 1,\\
y, & c = 0.
\end{cases}
\]
In quantum circuit descriptions, $c$ typically represents a conjunction of control bits, while $x$ and $y$ denote the selected quantum states or operations.
\end{itemize}

\subsection{Cost Metrics of Quantum Circuit}
The circuit's resource overhead is evaluated by its \textbf{Width} (total qubits required), \textbf{Depth} (sequential gate layers), and the number of \textbf{Ancilla} (\#Ancilla) qubits, which includes both \textbf{Garbage Ancilla} (not reset) and \textbf{Reusable Ancilla} (reset to $\ket{0}$). Beyond these fundamental metrics, we focus on the following non-Clifford costs:
\begin{itemize}
    \item \textbf{T-count}: The total number of $\TGATE$ and $\TGATE^\dagger$ gates ($\TGATE^\dagger$ is the inverse of $\TGATE$) appearing in the circuit. This metric quantifies the non-Clifford resource cost, as each $\TGATE$-type gate requires expensive magic-state distillation in a fault-tolerant implementation.
    \item \textbf{T-depth}: The minimal number of sequential layers of $\TGATE$ and $\TGATE^\dagger$ gates that must be executed, considering all possible parallelism across disjoint qubits. This metric reflects the time cost of non-Clifford operations.
\end{itemize}

\subsection{Quantum Variables and Kets}
For conciseness, we adopt the following notations for quantum registers and states throughout this paper:

\begin{itemize}
    \item \label{subsec:multiq-vars}
    An $n$-qubit register $X$ is written as $\ket{X}_n$ or $\ket{X_{n-1:0}} \equiv \ket{X_{n-1}\dots X_0}$, where $X_{n-1}$ and $X_0$ denote the MSB and LSB, respectively. The width subscript $n$ is omitted when it is clear from the context.

    \item 
    A state $\ket{0}$ of width $n$ represents the all-zero state $\ket{0}^{\otimes n}$.

    \item 
    Reusable ancilla qubits are denoted by $AN\ket{0}$, with multiple ancilla qubits indexed as $AN_1\ket{0}, AN_2\ket{0}, \dots$, indicating they are initialized to $\ket{0}$ before use.
\end{itemize}



\begin{figure}[t]
    \centering
    \includegraphics[width=0.85\linewidth]{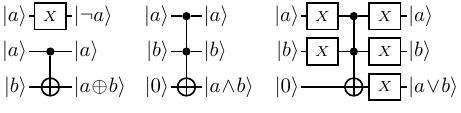}
    \caption{Quantum implementations of classical logic gates ($\GATE{NOT}$, $\GATE{XOR}$, $\GATE{AND}$, and $\GATE{OR}$).}
    \label{fig:Classical_gate_mapping}
\end{figure}

\subsection{Mapping Classical Logic Gates to Quantum Circuits} \label{subsec:Classical_gate_mapping}
Quantum computing has demonstrated the potential to solve problems intractable for classical computers. As a result, there has been growing interest in directly mapping classical Boolean logic onto quantum circuits~\cite{C_EDAQ}. This approach allows existing classical algorithms to be efficiently ported into the quantum domain. Standard Boolean primitives—including $\GATE{NOT}$, $\GATE{XOR}$, $\GATE{AND}$, and $\GATE{OR}$—are implemented via reversible constructions such as $\GATE{X}$, $\GATE{CCX}$ and $\GATE{CNOT}$ gates, typically requiring ancilla qubits to maintain unitarity~\cite{C_MapQ1, C_MapQ2}. 
As illustrated in Fig.~\ref{fig:Classical_gate_mapping}, this strategy facilitates the synthesis of quantum circuits from algebraic Boolean expressions through gate-by-gate emulation. Such reversible mappings underpin complex quantum arithmetic and oracle designs, allowing for the decomposition of high-level functional specifications into hardware-compatible quantum gate sequences.

\subsection{Toffoli ($\GATE{CCX}$) and Temporary Logical-AND ($\GATE{T\text{-}AND}$) Gate} \label{subsec:Temporary-logical-AND-gate}
Aiming for $\TGATE$–efficiency, we employ a Toffoli implementation guided by the meet-in-the-middle synthesis results of Amy et al.~\cite{AmyMaslovMosca2013MiM}, which characterize optimal-depth Clifford+T circuits and, in particular, enable the familiar 7–$\TGATE$, 3–$\TGATE$-depth, total depth~9 Toffoli decomposition (see Fig.~\ref{fig:Amy-TOF}).
\begin{figure}[tb]
    \centering
    \includegraphics[width=0.95\linewidth]{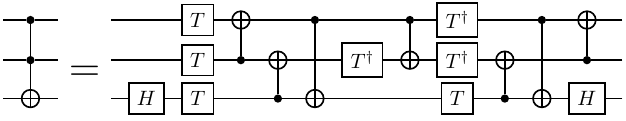}
    \caption{Toffoli gate~\cite{AmyMaslovMosca2013MiM} ($\TGATE^\dagger$ is the inverse of $\TGATE$)}
    \label{fig:Amy-TOF}
\end{figure}
\begin{figure}[t]
    \centering
    \includegraphics[width=0.95\linewidth]{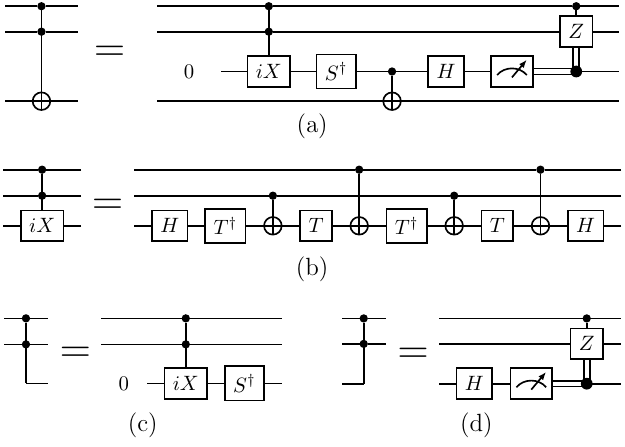}
    \caption{
    (a) Implementation of the Toffoli gate proposed by Jones~\cite{JonesToffoli}. 
    (b) Composition of the $\GATE{iX}$ gate~\cite{iXgate}. 
    (c) AND operation of the Temporary Logical-AND gate construction. 
    (d) Uncomputation using Hadamard and measurement.
    }
    \label{fig:temp_logical_AND}
\end{figure}

In order to reduce the T-count overhead of the Toffoli gate, Jones~\cite{JonesToffoli} proposed a more efficient implementation using only 4~T gates, as shown in Fig.~\ref{fig:temp_logical_AND}(a). This construction is based on an ancilla qubit and the so-called $\GATE{iX}$ gate, illustrated in Fig.~\ref{fig:temp_logical_AND}(b), which performs the operation:
\begin{equation}
\GATE{iX}(\ket{q_0}\ket{q_1}\ket{y}) = i^{q_0 q_1} \ket{q_0}\ket{q_1}\ket{y \oplus (q_0 q_1)}
\label{eq:ix}
\end{equation}
By applying $\GATE{iX}$ to an ancilla initialized to $\ket{0}$:
\begin{equation}
\GATE{iX}(\ket{q_0}\ket{q_1}\ket{0}) = i^{q_0 q_1} \ket{q_0}\ket{q_1}\ket{(q_0 q_1)}
\label{eq:ix_ancilla}
\end{equation}
Jones then corrects the phase by applying the $\GATE{S}^\dagger$ gate to the ancilla qubit. Once this value is copied to the target register, it is necessary to uncompute the $q_0 q_1$ state. However, Jones proposed using a Hadamard gate:
\begin{equation}
\GATE{H}\ket{q_0 q_1} = \frac{1}{\sqrt{2}} \sum_{z \in \{0,1\}} (-1)^{q_0 q_1 z} \ket{z}
\label{eq:hadamard}
\end{equation}
A measurement is then taken and the resulting phase is corrected with a $\GATE{CZ}$ gate.

Building upon this idea, Gidney~\cite{GidneyTempLogicalAND} introduced the Temporary Logical-AND gate ($\GATE{T\text{-}AND}$), which decomposes the Toffoli gate into two modular components: an AND operation and its uncomputation (see Fig.~\ref{fig:temp_logical_AND}(c)--(d)). Specifically, Gidney optimized the AND operation using a more $\TGATE$-efficient architecture (Fig.~\ref{fig:TAND}), which executes on a $\TGATE$-state ancilla $\ket{T}$ with $\TGATE$-count $=4$ and $\TGATE$-depth $=2$. The corresponding uncomputation (Fig.~\ref{fig:temp_logical_AND}(d)) remains unchanged and $\TGATE$-free. Therefore, this architecture minimizes the overall resource overhead.

\begin{figure}[t]
    \centering
    \includegraphics[width=0.95\linewidth]{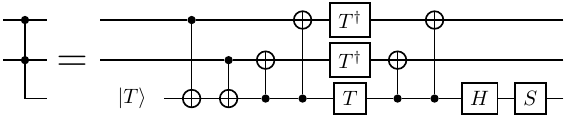}
    \caption{AND operation of the Temporary Logical-AND gate proposed by Gidney~\cite{GidneyTempLogicalAND}.}
    \label{fig:TAND}
\end{figure}

\subsection{Leading-One/Zero Counter (LOC/LZC)}\label{subsec:lzd}
For a binary word, the outputs of the leading-zero counter (LZC) and leading-one counter (LOC) are the numbers of consecutive `0's and `1's, respectively, from the most significant bit (MSB) to the first opposite bit~\cite{C_LZD,C_LZA_LZD,C_LZC,C_4LZC}.

\smallskip
\noindent\textbf{Mathematical definition.}
Let $X\in\{0,1,\dots,2^m-1\}$ be an $m$-bit unsigned integer, written as $X=\sum_{i=0}^{m-1} b_i2^i$ ($b_i\in\{0,1\}$ and $b_{m-1}$ is the MSB). The leading-zero count is defined as
\[
\MATHFN{LZC}(X)\!=\!
\begin{cases}
m, &X \!=\! 0, \\
\min \{ \gamma \!\in\! \{0,\dots,m\} \!\mid\! b_{m-1-\gamma} \!=\! 1 \}, &X \!\neq\! 0,
\end{cases} 
\]
and the leading-one count (LOC) is defined as 
\[
\MATHFN{LOC}(X)\!=\!
\begin{cases}
m, & \!\!\!X \!=\! 2^m{-}1, \\
\min \{ \gamma \!\in\! \{0,\dots,m\} \!\mid\! b_{m-1-\gamma} \!=\! 0 \}, & \!\!\!X \!\neq\! 2^m{-}1.
\end{cases}
\]
The two functions are related by bitwise complementation:
\begin{equation}\label{eq:LZCeqLOC}
\MATHFN{LZC}(X)=\MATHFN{LOC}(\overline{X}),\qquad
\MATHFN{LOC}(X)=\MATHFN{LZC}(\overline{X}).
\end{equation}
For the running example,
\(\MATHFN{LZC}\big((0000\,0000\,0010\,1011)_2\big)=\MATHFN{LOC}\big((1111\,1111\,1101\,0100)_2\big)=10.\) 
\section{Methodology}\label{sec:Methodology}
This section details the proposed methodology for constructing the QLZC. 
Rather than a direct mapping of classical counter boolean logic into a reversible format, we reformulate the counting process to better suit the requirements of reversible synthesis. 
Based on this formulation, we derive a set of functional primitives that serve as the modular building blocks for our architecture. 
Leveraging these primitives, we develop a modular sequential QLZC design, followed by optimizations targeting $\TGATE$-efficiency and qubit reuse. 
Finally, we extend this construction into a Parallel QLZC (PQLZC) architecture through a hierarchical merge strategy, which enhances scalability and further reduces circuit depth. 
Furthermore, given the minimal logical disparity between the QLZC and QLOC, the latter can be easily transformed into a QLZC—or even a reconfigurable version—by simply appending $\GATE{X}$ or $\GATE{CX}$ gates to the input stage. This flexibility demonstrates a "polymorphism" property within our design, a detailed discussion of which is presented in Sec.~\ref{subsec:Discussion}.

\subsection{Modular Logic Reformulation for QLZC Construction}\label{subsubsec:LZC_Algorithm}

The leading-zero counter (LZC)~\cite{C_LZD,C_LZA_LZD,C_LZC,C_4LZC} identifies the position of the first `1' in a bit string and outputs the corresponding number of preceding `0's. 
A straightforward implementation may therefore be constructed by repeatedly detecting the leading-zero condition and updating the output through addition. 
However, directly translating such a classical realization into a quantum circuit by gate-level substitution does not fully exploit the advantages of reversible quantum operations. 
Existing quantum leading-zero detector designs~\cite{Q_LZD,Q_Fdivision_LZD,Q_Fadder_LZD,Q_Fadder_LZC} commonly follow a Boolean-to-quantum mapping flow, where the output bits are first expressed as Boolean formulas and then simplified to reduce resource cost. 
While this approach can produce compact circuits for a fixed input width, the resulting Boolean expressions lack structural regularity and vary significantly with the input size. Consequently, any change in bit width necessitates re-deriving and re-optimizing the entire logic, hindering scalable and modular extension.

In this work, we instead exploit the conditional bit-flip capability naturally provided by quantum controlled gates such as $\GATE{CX}$, $\GATE{CCX}$, and $\GATE{MCX}$. 
Rather than relying on complex Boolean synthesis, we reformulate the update step as a sequence of condition-triggered bit-toggles on the output register. This reformulation aligns with the systematic bit-flip patterns of binary incrementing, leading to a more resource-efficient and modular construction.
%
Since the leading-one counter (LOC) admits a more direct detection condition, we first construct an LOC-oriented algorithm (Alg.~\ref{alg:QOLZC}) and prove its correctness in Theorem~\ref{thm:QOLZC}. 
The LZC circuit is then obtained from \eqref{eq:LZCeqLOC} by bitwise complementation of the input before performing LOC.

\begin{figure*}[tb]
    \centering
    \begin{subfigure}[t]{0.5\linewidth}
        \centering
        \includegraphics[width=\linewidth]{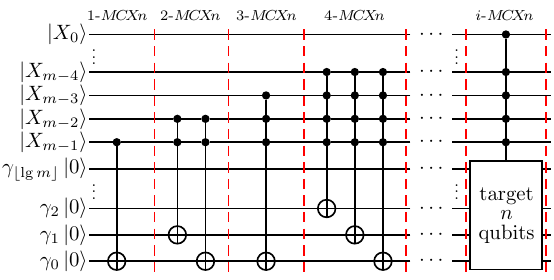}
        \caption{Original version.}
        \label{fig:iMCXn_gate}
    \end{subfigure}
    \hfill
    \begin{subfigure}[t]{0.19\linewidth}
        \centering
        \includegraphics[width=\linewidth]{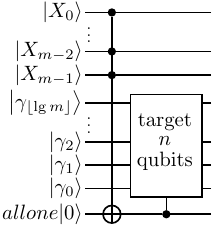}
        \caption{Modified version.}
        \label{fig:iMCXn_plus}
    \end{subfigure}
    \hfill
    \begin{subfigure}[t]{0.24\linewidth}
        \centering
        \includegraphics[width=\linewidth]{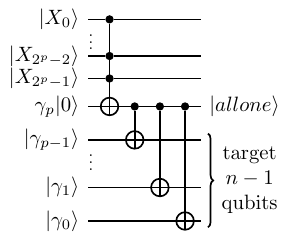}
        \caption{$2^p$ Optimized version.}
        \label{fig:MCXn_pow2}
    \end{subfigure}
    \caption{Quantum circuits of the $\GATE{i\text{-}MCXn}$ gate: (a) original version, (b) modified version utilizing an additional ancilla to store the $\GATE{MCX}$ result, and (c) optimized version for $i=2^p$.}
    \label{fig:iMCXn_all}
\end{figure*}

\begin{algorithm}[H]
\caption{Modular LOC (MLOC)}
\begin{algorithmic}[1]
  \STATE \textbf{Input:} $X$ — a \(m\)-bits binary array
  \STATE \textbf{Output:} $\gamma$ — a \((\lfloor \lg m\rfloor+1)\)-bits binary array, the leading one count of \(X\)

  \STATE \(\gamma \gets 0\)\quad \COMMENT{Initialize count}
    \STATE\hspace{0em}\textbf{for} \(i = 1\) to \(m\) \textbf{do }
    \STATE\hspace{1em}\textbf{if} \(X_{m\!-\!1}\!\land\!X_{m\!-\!2}\!\land\!\cdots\!\land\!X_{m\!-\!i}\) \textbf{then }
      \STATE\hspace{2em}\COMMENT{Update $\gamma$ from $i-1$ to $i$ by flipping}
      \STATE\hspace{2em}\(\Delta \leftarrow (i-1) \oplus i\) \COMMENT{Identify bits to be flipped}
      \STATE\hspace{2em}\textbf{for} each bit index \(j\) where \(\Delta_j = 1\) \textbf{do} flip \(\gamma_j\)

  \RETURN \(\gamma\)
\end{algorithmic}
    \label{alg:QOLZC}
\end{algorithm}

Specifically, for an $m$-bit input, the detection of an $i$-bit leading-one prefix can be expressed by the Boolean condition
\(
X_{m\!-\!1}\land X_{m\!-\!2}\land \cdots \land X_{m\!-\!i}.
\) 
Instead of updating the count register through explicit addition, the corresponding update in that round is realized by flipping the output bits that must change under this condition. 
Specifically, the bitwise difference between $i-1$ and $i$ is captured by $\Delta = (i-1) \oplus i$, where each bit $\Delta_j = 1$ indicates that the $j$-th bit of $\gamma$ must be flipped. According to Lemma~\ref{lem:bit_flip_logic} identifies that the bits to be flipped are precisely the $n$ LSBs of $\gamma$, where $n \in \mathbb{N}_1$ is the unique integer satisfying $i \bmod 2^n = 2^{n-1}$. 
Since all bit flips in the same round share the same control condition, they can be grouped into a systematic $\GATE{MCX}$-based construction. 
We package this operation as a reusable subcomponent, referred to as the $\GATE{i\text{-}MCXn}$ gate, which serves as a basic building block in the proposed circuit design.

\begin{lemma} \label{lem:bit_flip_logic}
For any $i \in \mathbb{N}_1$, the bitwise difference between $i-1$ and $i$ is characterized by $\Delta = (i-1) \oplus i$. This difference $\Delta$ consists of exactly $n$ LSBs being $1$ (i.e., $\Delta = 2^n - 1$), where n is the unique integer satisfying $i \bmod 2^n = 2^{n-1}$.
\end{lemma}
\begin{proof}
Write $i$ in binary as \( i = \sum_{j=0}^{\infty} b_j 2^j \), where \( b_j \in \{0, 1\} \). Let \( n-1 \) be the index of the first non-zero bit, i.e., \( b_0 = b_1 = \dots = b_{n-2} = 0 \) and \( b_{n-1} = 1 \); thus, \( i=(\dots b_n 1 0 \dots 0)_2 \). Truncating $i$ to its $n$ LSBs yields \( (10\dots0)_2 = 2^{n-1} \), which satisfies the remainder condition \( i \bmod 2^n = 2^{n-1} \).

To obtain $i-1$, subtracting $1$ from $i$ triggers a borrow propagation that turns the '1' at index \( n-1 \) into '0' and flips all \( n-1 \) trailing zeros into '1's, resulting in \( i-1 = (\dots b_n 0 1 \dots 1)_2 \). 
The bitwise difference between the two, obtained via ( $\Delta = i \oplus (i-1)$ ), reveals that the bits differ precisely at the lowest ( $n$ ) positions. This confirms that ( $\Delta = (0\dots01\dots1)_2$ ), where exactly ( $n$ ) LSBs are set to $1$, identifying the bits to be flipped.
\end{proof}

\begin{theoremrep}[Correctness of the MLOC]\label{thm:QOLZC}
Let $X = (b_{m-1}, b_{m-2}, \dots, b_0)$ be an $m$-bit binary array. Algorithm~\ref{alg:QOLZC} correctly computes the leading zero count $\gamma = \MATHFN{LZC}(X)$.
\end{theoremrep}
\begin{proof}
Let $L = \MATHFN{LOC}(X)$. The condition in Line~5 holds iff $i \le L$, as the first non-one bit $b_{m-1-L}$ triggers a failure for all $i > L$. Thus, the update block (Lines 6) is executed $L$ times. 
By Lemma~\ref{lem:bit_flip_logic}, each $i$-th iteration identifies the bitwise difference $\Delta = (i-1) \oplus i$, which corresponds to flipping exactly $n$ LSBs of the counter, where $n$ is the unique integer satisfying $i \bmod 2^n = 2^{n-1}$. This ensures that $\gamma$ is correctly updated from $i-1$ to $i$ at each step. Starting from the initial state $\gamma = 0$, the $L$ successive increments result in a final value of $\gamma$ that accurately represents the binary encoding of $L$. 
\end{proof}

\subsection{Definition of the \texorpdfstring{$\GATE{i\text{-}MCXn}$}{i-MCXn} gate}\label{subsubsec:DEF_iMCXn}
As previously discussed, each conditional bit-flip iteration is structured into a systematic $\GATE{MCX}$-based construction, denoted as the $\GATE{i\text{-}MCXn}$ gate shown in Fig.~\ref{fig:iMCXn_gate}. According to Lemma~\ref{lem:bit_flip_logic}, this gate models one update step of the proposed counting procedure by flipping the $n$ LSBs of the target state $\ket{\gamma}$ if and only if $\bigwedge_{t=m-i}^{m-1} X_t = 1$, where $n$ is the unique integer satisfying $i \bmod 2^n = 2^{n-1}$; otherwise, it acts as the identity:
\[
\begin{aligned}
&\GATE{i\text{-}MCXn}\!\left(\ket{X_{m-1:m-i}}\!\ket{\gamma}\right) \\
&\!=\!
\ket{X_{m-1:m-i}}
\!\otimes\!
\MATHFN{ITE}\!\left(
\wedge_{t=m-i}^{m-1} X_t
,
\ket{\gamma_{\lfloor\lg m\rfloor:n}}
\ket{\overline{\gamma_{n-1:0}}},
\ket{\gamma}
\right) \\
&\!=\!
\ket{X_{m-1:m-i}}
\ket{\gamma_{\lfloor\lg m\rfloor:n}}
\otimes_{t=0}^{n-1}
\GATE{MCX}\!\left(
\ket{X_{m-1:m-i}}
,\;\ket{\gamma_t}
\right).
\end{aligned}
\]

This gate is realized by multiple $\GATE{MCX}$ gates sharing a common control condition, with each gate targeting a distinct qubit among the lowest $n$ bits.
Furthermore, per Lemma~\ref{lem:bit_flip_logic}, at round $i$ where $\ket{\gamma}=\ket{i-1}$ (equivalently, the prefix bits $X_{m-1}, \dots, X_{m-i-1}$ are all $1$), the $\GATE{i\text{-}MCXn}$ gate matches an increment-by-one guarded by the carry condition. Specifically, if $X_{m-i}=1$ such that $\bigwedge_{t=m-i}^{m-1} X_t = 1$, the gate increments the target register by one ($\ket{i-1}\to\ket{i}$); otherwise, it acts as the identity:
\begin{equation}\label{eq:iMCXn_add1}
\begin{aligned}
&\GATE{i\text{-}MCXn}\!\left(
\ket{X_{m-1:m-i}}
\ket{\gamma}\right) \\
&=
\ket{X_{m-1:m-i}}
\otimes
\MATHFN{ITE}\!\left(
\wedge_{t=m-i}^{m-1} X_t,\;
\ket{\gamma+1},\;
\ket{\gamma}
\right) \\
&=
\ket{X_{m-1:m-i}}
\otimes
\MATHFN{ITE}\!\left(
\wedge_{t=m-i}^{m-1} X_t,\;
\ket{i},\;
\ket{\gamma}
\right).
\end{aligned}
\end{equation}

However, the depth of multiple $\GATE{MCX}$ gates scales rapidly with increasing control counts. To mitigate this, we utilize an additional ancilla to store the $\GATE{MCX}$ result, subsequently triggering $\GATE{CX}$ gates to flip the remaining target qubits (Fig.~\ref{fig:iMCXn_plus}). This strategy significantly reduces overall depth and enhances architectural scalability.
As shown in Fig.~\ref{fig:MCXn_pow2}, for the power-of-two case $i = 2^p$
, the $\GATE{i\text{-}MCXn}$ gate targets qubits from index 0 to $p$. Since preceding stages ($i < 2^p$) only utilize qubits up to index $p-1$, the $p$-th qubit can be directly used to store the all-one control result and subsequently trigger the flips of the remaining $n-1$ target qubits. This in-place register recycling effectively eliminates the need for an external ancilla.

\subsection{Quantum Leading One Counter and Leading Zero Counter}
\begin{figure}[tb]
    \centering
    \includegraphics[width=\linewidth]{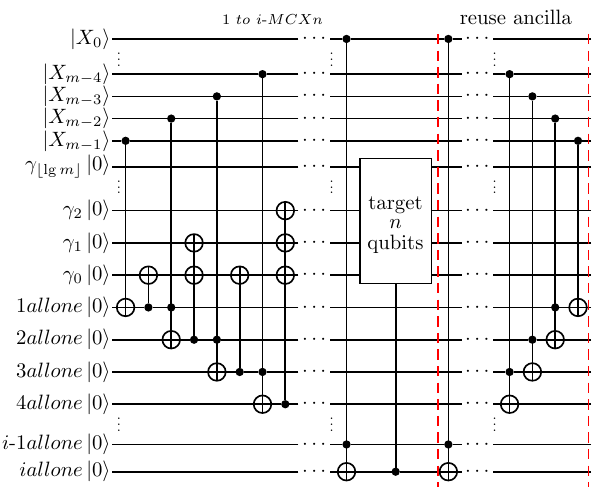}
    \caption{Quantum circuit of a Leading One Counter (QLOC)}
    \label{fig:MineLOC}
\end{figure}

\begin{figure*}[tb]
    \centering
    \begin{subfigure}[t]{0.32\linewidth}
        \centering
        \vspace{-85 pt}
        \includegraphics[width=\linewidth]{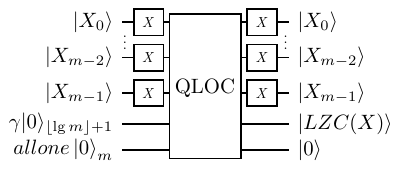}
        \caption{Original QLZC}
        \label{fig:MineLZC}
    \end{subfigure}
    \hfill
    \begin{subfigure}[t]{0.67\linewidth}
        \centering
        \includegraphics[width=\linewidth]{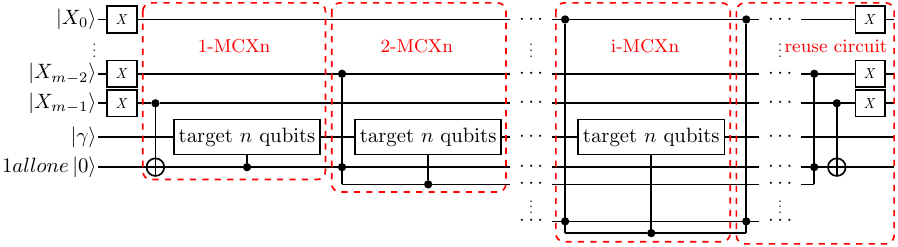}
        \vspace{-20 pt}
        \caption{TA-OP QLZC}
        \label{fig:mixed_Nop_LZC}
    \end{subfigure}
    \caption{Quantum Leading Zero Counter (QLZC). (a) Original QLZC. (b) TA-OP QLZC optimized with $\GATE{T\text{-}AND}$ gates.}
    \label{fig:qlzc_compare}
\end{figure*}

As illustrated in Fig.~\ref{fig:MineLOC}, the proposed Quantum Leading-One Counter (QLOC) is constructed through a sequential application of $\GATE{i\text{-}MCXn}$ gates for $i = 1, \dots, m$, and its correctness is proved in Theorem~\ref{thm:QLOC_correctness}. 
Let $X$ denote the $m$-qubit input, and let $\ket{\gamma}$ be an accumulator register of $(\lfloor \lg m \rfloor + 1)$ qubits initialized to $\ket{0}$. 
Per Eq.~\eqref{eq:iMCXn_add1}, each $\GATE{i\text{-}MCXn}$ gate increments $\gamma$ if and only if the $i$ most significant input bits are all ones (i.e., $\bigwedge_{t=m-i}^{m-1} X_t = 1$). 
To minimize depth and complexity, we employ incremental construction for the control logic. Instead of full multi-controlled operations, the $i$-th all-one flag is derived from the $(i-1)$-stage result and $X_{m-i}$ via a single $\GATE{CCX}$ gate:
\[
\begin{aligned}
&\GATE{MCX}(
\ket{X_{m-1:m-i}}
\ket{0})\\
&= 
\ket{X_{m-1:m-i}}
\ket{0 \oplus \wedge_{t=m-i}^{m-1} X_t} = 
\ket{X_{m-1:m-i}}
\ket{i\text{allone}} \\
&= 
\ket{X_{m-1:m-i}}
\ket{0 \oplus (0 \oplus \wedge_{t=m-i-1}^{m-1} X_t)(X_{m-i})} \\
&= 
\ket{X_{m-1:m-i+1}}
\GATE{CCX}\big(\ket{X_{m-i}}\ket{i{-}1\text{allone}}\ket{0}\big).
\end{aligned}
\]
This recursive scheme replaces repeated $\GATE{MCX}$ operations with a Toffoli ladder. To reclaim ancillas, we uncompute these flags by reversing the $\GATE{CCX}$ sequence. Notably, such uncomputation is not directly applicable for the $i = 2^p$ optimization (Fig.~\ref{fig:MCXn_pow2}).

\begin{theoremrep}[Correctness of the QLOC]\label{thm:QLOC_correctness}
The sequential circuit construction of $\GATE{i\text{-}MCXn}$ gates correctly implements the Leading-One Detection function: 
\[
\CKTFN{QLOC}\big(\ket{X}_m\ket{0}_{\lfloor\lg m\rfloor+1}\big)
\;=\;
\ket{X}\ket{\MATHFN{LOC}(X)}
\;=\;
\ket{X}\ket{\gamma}\,.
\]
\end{theoremrep}
\begin{proof}
The proof establishes a direct mapping between the QLOC circuit and the MLOC algorithmic logic. By Eq.~\eqref{eq:iMCXn_add1}, the $i$-th $\GATE{i\text{-}MCXn}$ gate executes the $i$-th conditional update block: mapping $\ket{\gamma} \mapsto \ket{i}$ if the prefix condition $\bigwedge_{t=m-i}^{m-1} X_t$ holds, and $\ket{\gamma} \mapsto \ket{\gamma}$ otherwise. 
Since the $\GATE{i\text{-}MCXn}$ gate specifically decomposes into the bit-flipping sequence as described by Lemma~\ref{lem:bit_flip_logic}, its hardware execution is functionally equivalent to the iterative process in Algorithm~\ref{alg:QOLZC}. Given that these conditional increments are triggered sequentially from $i=1$ and cease at the first occurrence of $X_{m-i}=0$, the cumulative value stored in the register $\gamma$ is guaranteed to represent the total number of prefix ones. Consequently, the correctness of the QLOC circuit is established by its direct correspondence to the proven MLOC construction.
\end{proof}

Finally, the QLZC is constructed by sandwiching the QLOC circuit between layers of $\GATE{X}$ gates (Fig.~\ref{fig:MineLZC}). This configuration leverages the duality $\MATHFN{LZC}(X) = \MATHFN{LOC}(\bar{X})$, where $\ket{\bar{X}}$ is the bit-wise inversion of $\ket{X}$. The resulting transformation is:
\[
\begin{aligned}
&\CKTFN{QLZC} \big( \ket{X}_m \ket{0}_{\lfloor\lg m\rfloor+1} \big) \\
&= (\GATE{X}^{\otimes m} \otimes \GATE{I}) \cdot \CKTFN{QLOC} \cdot (\GATE{X}^{\otimes m} \otimes \GATE{I}) \left( \ket{X}_m \ket{0} \right) \\
&= (\GATE{X}^{\otimes m} \otimes \GATE{I}) \cdot \CKTFN{QLOC} \left( \ket{\bar{X}}_m \ket{0} \right) \\
&= (\GATE{X}^{\otimes m} \otimes \GATE{I}) \left( \ket{\bar{X}}_m \ket{\MATHFN{LOC}(\bar{X})} \right) 
= \ket{X}_m \ket{\MATHFN{LZC}(X)}.
\end{aligned}
\]

\subsection{Optimization via Temporary Logical-AND ($\GATE{T\text{-}AND}$).}\label{subsubsec:TAND}
In fault-tolerant quantum computing, $\TGATE$ gates are significantly more resource-intensive than Clifford gates. Consequently, minimizing both $\TGATE$-count and $\TGATE$-depth is critical for achieving a $\TGATE$-efficient design. 
To this end, by incorporating the concept of $\GATE{T\text{-}AND}$ utilized in~\cite{Q_LZD, Q_Fdivision_LZD}(see Sec.~\ref{subsec:Temporary-logical-AND-gate}), we group each all-zero detection circuit with its corresponding reuse logic and implement them using $\GATE{T\text{-}AND}$, optimizing the overall quantum circuit structure, as shown in Fig.~\ref{fig:mixed_Nop_LZC}. For the $i = 2^p$ $\GATE{i\text{-}MCXn}$ gate, the same optimization can be applied. 
Such designs employing this technique are labeled “TA-OP” to distinguish them from conventional implementations.

\begin{figure*}[bt]
    \centering
    \begin{minipage}[t]{0.64\linewidth}
    \vspace{0pt}
    \centering
        \begin{subfigure}[t]{0.485\linewidth}
            \centering
            \includegraphics[width=\linewidth]{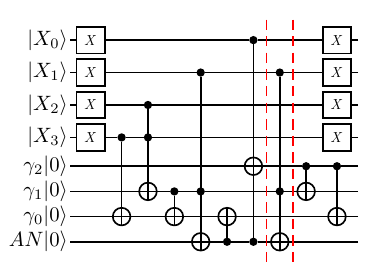}
            \caption{P-OP 4-QLZC}
            \label{fig:Mine4LZC_2}
        \end{subfigure}
        \hfill
        \begin{subfigure}[t]{0.40\linewidth}
            \centering
            \includegraphics[width=\linewidth]{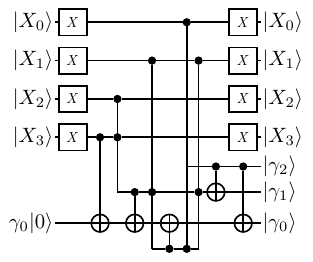}
            \caption{TA-P-OP 4-QLZC}
            \label{fig:mixed_Op_4LZC}
        \end{subfigure}
        \caption{4-qubit QLZC. (a) 4-QLZC utilizing optimized $\GATE{2^p\text{-}MCXn}$ gates, where the red lines indicate the reused ancilla. (b) further optimized with $\GATE{T\text{-}AND}$ gates.}
        \label{fig:4bit_qlzc}
    \end{minipage}
    \hfill 
    \begin{minipage}[t]{0.28\linewidth}
    \vspace{4.6pt}
        \centering
        \includegraphics[width=\linewidth]{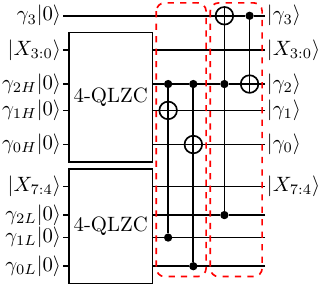}
        \caption{8-qubit PQLZC constructed by the merge circuit}
        \label{fig:8PLZC}
    \end{minipage}
\end{figure*}

\subsection{Parallelized Quantum Leading Zero Counter (PQLZC) }
Although our proposed circuit already significantly improves upon existing designs, its depth still increases linearly with the number of input qubits. Therefore, parallelization is essential. 
We adopt a strategy similar to that of~\cite{Q_LZD}, using a 4-qubit grouped architecture to construct our PQLZC. 
First, We construct a 4-bit QLZC utilizing optimized $\GATE{2^p\text{-}MCXn}$ gates as shown in Fig.~\ref{fig:Mine4LZC_2}. Notably, the $\GATE{3\text{-}MCXn}$ reuse logic must precede the $\GATE{4\text{-}MCXn}$ flip operation, as the target qubits of the latter overlap with the control lines of the former. Integrating the $\GATE{T\text{-}AND}$ method (Sec.~\ref{subsubsec:TAND}) further refines this architecture into the $\GATE{T\text{-}AND}$ based 4-QLZC design illustrated in Fig.~\ref{fig:mixed_Op_4LZC}.

Next, we utilize a merge circuit to construct an 8-qubit QLZC, a process that can be recursively extended to larger input sizes as illustrated in Fig.~\ref{fig:8PLZC} and \ref{fig:TA_PLZC}, and prove its correctness in Theorem~\ref{thm:QLZC_merge_correctness}.  
Since the QLZC is designed to count leading zeros, the low-order QLZC value is added only when the high-order subcircuit detects an all-zero input. While this behavior could be implemented using a controlled adder, we optimize the circuit by leveraging the specific property of an $m$-qubit QLZC (where $m=2^p$): if the input is all-zero, the output is $\gamma = m$ (binary $10\dots0$), meaning only the MSB is 1 (as the all-zero flag). Consequently, the controlled addition can be simplified using only $\lg(m)+1$ Toffoli ($\GATE{CCX}$) gates and one $\GATE{CX}$ gate, significantly reducing the required resources. 
Finally, for the output bit $\gamma_{\log_2(2m)}$, we employ the $\GATE{T\text{-}AND}$ (optimized Toffoli) gate to refine the datapath performance.

\begin{theoremrep}[Correctness of the QLZC Merge Circuit]\label{thm:QLZC_merge_correctness}
The recursive construction using the merge circuit correctly implements the Quantum Leading Zero Detection function. Specifically, for an $2m$-qubit input, the circuit realizes:
\[
\begin{gathered}
\CKTFN{PQLZC}\big(\ket{X}_{2m}\ket{0}_{\lg(2m)+1}\big)
\!=\!
\ket{X}\ket{\gamma}\!,\;
\\
\ket{\gamma}
\!=\!\ket{\MATHFN{ITE}(\gamma_{MSBH},\gamma_H \!+\! \gamma_L,\gamma_H)} 
\!=\!\ket{\MATHFN{LZC}(X)}\!,
\end{gathered}
\]
in which $\gamma_{H}$ and $\gamma_{L}$ are the outputs of the high- and low-order $m$-qubit QLZC subcircuits, respectively, and $\gamma_{MSBH}$ is the most significant bit of $\gamma_{H}$.
\end{theoremrep}
\begin{proof}
The construction follows from the recursive definition of the LZC function. For a $2m$-qubit input $X$, let $X_H$ and $X_L$ be the high- and low-order $m$-qubit halves, respectively. The total leading zero count $\gamma$ is $m + \gamma_L$ if $X_H = 0$, and $\gamma_H$ otherwise. 
In our design, the $m$-qubit subcircuit satisfies a boundary property where output $\gamma_{H(L)}=\ket{m}$ is encoded as $\ket{10\dots0}$ (only MSB $\gamma_{MSBH(L)} = 1$).
Consequently, the arithmetic addition $\gamma_H + \gamma_L$ simplifies to a controlled selection ($\MATHFN{ITE}$) logic applied across the output registers, as addition is only non-trivial when $\gamma_{MSBH}=1$:
\[ 
\begin{aligned}
\ket{\gamma} 
&=\ket{\MATHFN{ITE}(\gamma_{MSBH},\gamma_H\!+\!\gamma_L,\gamma_H)} 
\\&= \ket{\MATHFN{ITE}(\gamma_{MSBH}, \gamma_{MSBH}\!+\!\gamma_{MSBL}, \gamma_{MSBH})} \\&\quad\otimes \ket{\MATHFN{ITE}(\gamma_{MSBH}, 0\!+\!\gamma_{QL}, \gamma_{QH})},
\end{aligned}
\]
in which $\gamma_{MSBH,L}$ denote the MSBs, and $\gamma_{QH,L}$ denote the sub-registers of remaining bits. The implementation is verified in two functional blocks:

\smallskip\noindent \textbf{a) MSB and Overflow Logic ($\gamma_{\lg m}$ and $\gamma_{\lg 2m}$):}
The two highest bits of the output handle the arithmetic carry, implementing the transformation $\ket{\gamma_{\lg 2m}}\ket{\gamma_{\lg m}} = \ket{\MATHFN{ITE}\;(\gamma_{\lg mH}, \gamma_{\lg mH} + \gamma_{\lg mL}, \gamma_{\lg mH})}$. 
First, we define $\gamma_{\lg 2m} = \gamma_{\lg mH} \land \gamma_{\lg mL}$, which corresponds to the case where both subcircuits report an all-zero input. For the qubit $\gamma_{\lg m}$, the selection logic can be expanded as:
\[
\begin{aligned}
\gamma_{\lg m} &= \MATHFN{ITE}(\gamma_{\lg mH},\ \gamma_{\lg mH} + \gamma_{\lg mL},\ \gamma_{\lg mH}) 
\\&= \MATHFN{ITE}(\gamma_{\lg mH},\ \neg \gamma_{\lg mL},\ \gamma_{\lg mH})
\\&= \MATHFN{ITE}(\gamma_{\lg mH} \land \gamma_{\lg mL},\ \neg\gamma_{\lg mH},\ \gamma_{\lg mH}) 
\\&= \MATHFN{ITE}(\gamma_{\lg 2m},\ \neg\gamma_{\lg mH},\ \gamma_{\lg mH}),
\end{aligned}
\]
so $\gamma_{\lg m}=\gamma_{\lg 2m} \oplus \gamma_{\lg mH}$. 
Thus, the logic is realized by a $\GATE{CCX}(\gamma_{\lg m H}, \gamma_{\lg m L}, \gamma_{\lg 2m})$ gate followed by a $\GATE{CX}(\gamma_{\lg 2m}, \gamma_{\lg m H})$ gate. 
The $\GATE{CCX}$ gate is further optimized using a $\GATE{T\text{-}AND}$ operation to reduce resource overhead. 

\smallskip\noindent \textbf{b) Register-wide Selection Logic ($\gamma_Q$):}
For the remaining bits $i \in \{1, \dots, \lg m\!-\!1\}$, the merge circuit implements the conditional selection $\gamma_i = \MATHFN{ITE}(\gamma_{MSBH}, \gamma_{iL}, \gamma_{iH})$. 
Applying a $\GATE{CCX}$ to $\ket{\gamma_{MSBH}}\ket{\gamma_{iL}}\ket{\gamma_{iH}}$ yields the transformation $\gamma_{iH} \leftarrow \gamma_{iH} \oplus (\gamma_{MSBH} \land \gamma_{iL})$. We verify this implementation by expanding the $\GATE{XOR}$-$\GATE{AND}$ logic:
\[
\begin{aligned}
&\gamma_{iH} \oplus (\gamma_{MSBH} \land \gamma_{iL}) \\&= \MATHFN{ITE}(\gamma_{MSBH} \land \gamma_{iL},\ \neg \gamma_{iH},\ \gamma_{iH}) \\
&= \MATHFN{ITE}(\gamma_{MSBH},\ \MATHFN{ITE}(\gamma_{iL},\ \neg \gamma_{iH},\ \gamma_{iH}),\ \gamma_{iH}).
\end{aligned}
\]
By the all-zero property: if $\gamma_{MSBH} = 1$, the input is all-zero, forcing $\gamma_{QH} = 0$. Thus, $\MATHFN{ITE}(\gamma_{iL}, \neg \gamma_{iH}, \gamma_{iH})$ simplifies to $\MATHFN{ITE}(\gamma_{iL}, 1, 0) = \gamma_{iL}$ when $\gamma_{MSBH}$ is high, yielding:
\begin{equation}\label{eq:QH}
\gamma_{iH} \oplus (\gamma_{MSBH} \land \gamma_{iL}) = \MATHFN{ITE}(\gamma_{MSBH},\ \gamma_{iL},\ \gamma_{iH}).
\end{equation}
Therefore, sequential $\GATE{CCX}$ gates applied across the sub-registers correctly implement the multiplexing logic required for the recursive QLZC merge.
\end{proof}

\begin{figure}[bt]
    \centering
        \begin{subfigure}[t]{\linewidth}
        \centering
        \includegraphics[width=0.8\linewidth]{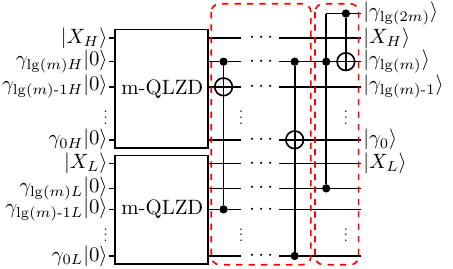}
        \caption{TA-OP PQLZC}
        \label{fig:TA_PLZC}
    \end{subfigure}
    \hfill
    \begin{subfigure}[t]{\linewidth}
        \centering
        \includegraphics[width=0.95\linewidth]{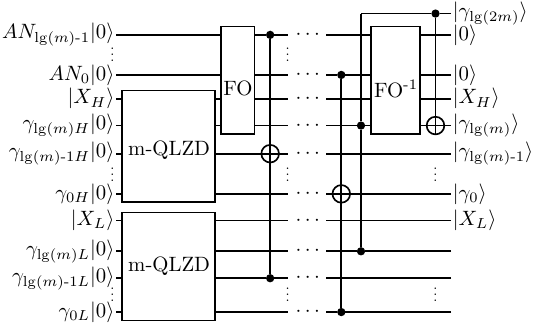}
        \caption{FO-TA-OP PQLZC}
        \label{fig:FO_PLZC}
    \end{subfigure}
    \caption{$2m$-qubit PQLZC ($m = 4 \cdot 2^p$). Internal $m$-QLZC qubits omitted for brevity. (a) optimized with $\GATE{T\text{-}AND}$ gates. (b) further optimized with fanout control.}
    \label{fig:mixed_PLZC}
\end{figure}

Moreover, all $\GATE{CCX}$ and $\GATE{T\text{-}AND}$ gates in the merge circuit share only the same control qubit, \(\gamma_{\lg(m)-1,H}\), whereas the remaining control and target qubits are mutually exclusive. Hence, this common control qubit can be fanned out to ancilla so that these gates can be executed in parallel, reducing the $\TGATE$-depth to a constant, i.e., the $\TGATE$-depth of a single $\GATE{CCX}$ gate, as shown in Fig.~\ref{fig:FO_PLZC}. We adopt a standard binary-tree fan-out subroutine to copy a control qubit to \(n\) ancilla qubits in logarithmic depth (see Alg.~\ref{alg:FO}). 
In addition, the optimized \(2m\)-PQLZC (\(m = 4 \cdot 2^p
\)) requires
\(
\frac{2m}{8}\cdot 2=\mathcal{O}(2m)
\)
additional fan-out ancilla. In the first merge round, the \(4\)-QLZC blocks are merged pairwise, yielding \(2m/8\) merge groups, and each group requires three identical control qubits in total, i.e., the original one and two ancilla copies. In the \(r\)-th merge round, each group requires \(r+2\) identical control qubits in total, so \(r+1\) ancilla are needed, while the number of groups decreases to \(2m/2^{r+2}\). Hence, the required number of additional ancilla in round \(r\) is
\(
\frac{2m}{2^{r+2}}(r+1),
\)
which decreases as \(r\) increases. Therefore, the first merge round requires the largest number of fan-out ancilla. 
Furthermore, with the reverse fan-out circuit \(\mathrm{FO}^{-1}\), these ancilla can be uncomputed and reused in every round. Therefore, the ancilla required in the first round also determine the total number of additional ancilla for the whole circuit. 
The only difference from Theorem~\ref{thm:QLZC_merge_correctness} is that, in \eqref{eq:QH}, the control qubit is replaced by the fan-out ancilla. That is,
\[
\gamma_{iH} \oplus (AN_i \land \gamma_{iL})
= \MATHFN{ITE}(\gamma_{MSBH},\gamma_{iL},\gamma_{iH}).
\]
Since \(\mathrm{FO}\) ensures \(AN_i=\gamma_{MSBH}\) for all \(i\), the above equality still holds. Moreover, the reverse \(\mathrm{FO}\) is applied after the $\GATE{T\text{-}AND}$ gate for \(\gamma_{\lg 2m}\) is executed. Hence, the fan-out operation does not affect the behavior of the original circuit.

\begin{algorithm}[H]
\caption{Fan out (FO)}
\label{alg:FO}
\begin{algorithmic}[1]
\STATE \textbf{Input:} $1$-qubit control $\ket{ctrl}$; $n$-qubit ancilla $AN\ket{0}_n$; 
\STATE \textbf{Output:}  $AN'=\bigotimes_{a=n-1}^{0}\ket{ctrl}$
\STATE \(Q \gets [AN_{n-1}, \dots, AN_1, AN_0,ctrl ]\)
\STATE \hspace{1em}\textbf{for } $i = 0,1,\dots,\,\lceil \lg (n{+}1) \rceil - 1$ \textbf{do }
\STATE \hspace{2em}\textbf{for } $j = 0,1,\dots,\,2^i - 1$ \textbf{do }
\STATE \hspace{3em}\textbf{if } $j + 2^i \le n$ \textbf{then }$\GATE{CX}$\((Q_j, Q_{j+2^i})\)
\RETURN $AN'$
\end{algorithmic}
\end{algorithm}

\begin{lemmarep}[Standard property of FO]
\label{lem:fanout}
Given one control \(\ket{ctrl}\) and \(n\)-qubits ancilla $\ket{AN}$ initialized to \(\ket{0}\), FO copies \(\ket{ctrl}\) to every \(AN_\ell\) for \(0\le \ell\le n-1\).
The depth of FO is \(\lceil \lg(n+1)\rceil\).
\end{lemmarep}
\begin{proof}
Let \(Q=(Q_n,\dots,Q_0)=(AN_{n-1},\dots,AN_0,ctrl)\), so initially only \(Q_0\) stores \(\ket{ctrl}\). 
This follows from the standard doubling property of binary-tree fan-out.
In round \(i\), each gate \(\GATE{CX}(Q_j,Q_{j+2^i})\) copies \(\ket{ctrl}\) from an existing qubit to a fresh one, so the number of qubits storing \(\ket{ctrl}\) doubles, up to \(n+1\).
Hence, after \(\lceil \lg(n+1)\rceil\) rounds, all ancilla qubits store \(\ket{ctrl}\).
Since all $\GATE{CNOT}$ gates within the same round are disjoint, each round has depth \(1\), and thus the total depth is \(\lceil \lg(n+1)\rceil\).
\end{proof}

\subsection{Discussion on Functional Polymorphism of QLZOC}\label{subsec:Discussion}
Although the subsequent sections primarily focus on the implementation and evaluation of the QLZC, the proposed architectural framework is inherently applicable to the QLOC designs. As illustrated in Fig.~\ref{fig:qlzc_compare}, the structural difference between these configurations is minimal: implementing a QLZC using the QLOC core logic merely requires a bitwise $\GATE{X}$ gate layer for input complementation, followed by an identical uncomputation step after the leading-zero position is identified. Conversely, a native QLOC implementation eliminates these negation overheads, reflecting the functional polymorphism of our architectural design. 
Consequently, the resource complexities and optimization strategies derived in this work—such as hierarchical merging and fan-out control—remain applicable to both configurations without loss of generality.

This inherent functional polymorphism allows for a reconfigurable version of the QLZOC to be straightforwardly implemented by replacing static $\GATE{X}$ gates with $\GATE{CX}$ gates, all triggered by a single mode-selection qubit. This modification allows the circuit to dynamically toggle between zero and one detection modes during runtime, providing a versatile functional primitive for complex quantum arithmetic pipelines.

\begin{table}[b]
\centering
\begin{tabular}{l c c c c c c c}
\hline
\textbf{Circuit} & \!\!\textbf{$\TGATE$-count\!\!} & \!\!\textbf{$\TGATE$-depth}\!\! & \!\!\textbf{\#Ancilla}\!\!  & \!\!\textbf{Width}\!\!  & \!\!\textbf{Depth}\!\!\\\hline
\!\cite{Q_LZD} circuit 1 (c1) & 24       & 12       & 3        & 6        & 64       \\
\!\cite{Q_LZD} circuit 2 (c2) & 20       & 10       & 2        & 5        & 61       \\
\!\cite{Q_LZD} circuit 3 (c3) & 33       & 15       & 1        & 4        & 67       \\\hline \hline
P-OP 4-QLZC                  & 28       & 12       & 1        & 4        & 42       \\
\% imp. over c2              & \!\!\TCfoVtw\!\! & \!\!\TDfoVtw\!\! & \!\!\AfoVtw\!\!  & \!\!20.00\%\!\!  & \!\!\DfoVtw\!\!  \\
\% imp. over c3              & \!\!\TCfoVtr\!\! & \!\!\TDfoVtr\!\! & \!\!\AfoVtr\!\!  & \!\! 0.00\%\!\!  & \!\!\DfoVtr\!\!  \\\hline \hline
TA-P-OP 4-QLZC               & 12       & 4        & 1        & 4        & 32       \\
\% imp. over c2              &\!\!\TCTfoVtw\!\! &\!\!\TDTfoVtw\!\! & \!\!\ATfoVtw\!\! & \!\!20.00\%\!\!  & \!\!\DTfoVtw\!\! \\
\% imp. over c3              &\!\!\TCTfoVtr\!\! &\!\!\TDTfoVtr\!\! & \!\!\ATfoVtr\!\! & \!\! 0.00\%\!\!  & \!\!\DTfoVtr\!\! \\\hline \hline
TA-OP 4-QLZC                 & 12       & 4        & 4        & 7        & 47       \\
\% imp. over c2              &\!\!\TCTNfoVtw\!\!&\!\!\TDTNfoVtw\!\!& \!\!\ATNfoVtw\!\!&\!\!-40.00\%\!\!  & \!\!\DTNfoVtw\!\!\\
\% imp. over c3              &\!\!\TCTNfoVtr\!\!&\!\!\TDTNfoVtr\!\!& \!\!\ATNfoVtr\!\!&\!\!-75.00\%\!\!  & \!\!\DTNfoVtr\!\!\\
\hline   
\end{tabular}
\caption{Comparison of $4$-qubit input QLZC designs.}
\label{tab:4bit_comparison}
\end{table}

\begin{table}[b]
\centering
\begin{tabular}{l c c c c c}
\hline
\textbf{Circuit} & \!\!\textbf{$\TGATE$-count}\!\! & \!\!\textbf{$\TGATE$-depth}\!\! & \!\!\textbf{\#Ancilla}\!\! & \!\!\makecell{\textbf{\#Garbage}\\\textbf{Ancilla}}\!\! & \!\!\textbf{Depth}\!\!\\
\hline
\!\cite{Q_LZD} circuit 1   & 48+18 & 12+7 & 9 & 3 & 64+23 \\
\!\cite{Q_LZD} circuit 2   & 40+18 & 10+7 & 7 & 3 & 61+23 \\
\!\cite{Q_LZD} circuit 3   & 66+18 & 15+7 & 5 & 3 & 67+23 \\
TA-OP 8-QLZC\!\!          & 28    & 8    & 8 & 0 & 103   \\
TA-OP 8-PQLZC\!\!          & 42    & 11   & 5 & 3 & 59    \\
FO-TA-OP 8-PQLZC\!\!          & 42    & 7   & 7 & 3 & 46    \\
\hline
\end{tabular}
\caption{Comparison of $8$-qubit input QLZC designs. We adopted our merge circuit for both designs and indicated the extra overhead of the merge circuit using a "+".}
\label{tab:8bit_comparison}
\end{table}

\section{Related Work} \label{sec:RelatedWork}
Early quantum leading zero counter or detector were often embedded in floating-point normalization datapaths, such as the single-precision adder~\cite{Q_Fadder_LZC, Q_Fadder_LZD}. Subsequent resource-aware designs for division \cite{Q_Fdivision_LZD} and optimized work (Orts \emph{et al.}, 2023) leverage classical Boolean formulations for fault-tolerant metrics by employing SAT-based synthesis and $\GATE{T\text{-}AND}$ gates~\cite{Q_LZD}.

Compared to prior designs~\cite{Q_LZD,Q_Fdivision_LZD,Q_Fadder_LZD,Q_Fadder_LZC} under the Clifford+T gate set $\{\GATE{H}, \GATE{S}, \GATE{CNOT}, \GATE{T}\}$, our work features:
\begin{enumerate}
  \item \textbf{Scalable Modular Design}: Unlike Boolean-derived circuits that require width-dependent re-optimization of $\GATE{CCX}$ gates and ancilla, our sequential construction employs a systematic structural template that scales linearly at a cost of $1$ $\GATE{T\text{-}AND}$ gate per input qubit. Our architecture eliminates manual re-tuning and writes the leading-zero count directly onto qubits, providing a seamless interface for downstream normalization and shifting.

  \item \textbf{Optimized Parallelization and Fan-Out Control}: Our \emph{parallelized} variant introduces a dedicated merge circuit that exploits a unique property: an all-zero $m$-qubit block yields $\gamma=m$ (binary $10\ldots 0$), where $m=2^p$. This allows us to substitute controlled adders with mere $(\lg(m)\!+\!1)$ $\GATE{CCX}$ $+1$ $\GATE{CX}$ gates. 
  Moreover, with logarithmic-depth fan-out control, these $\GATE{CCX}$ gates can be executed in parallel, reducing the $\TGATE$-depth of merge circuit to a constant. For an \(m\)-qubit PQLZD, this optimization requires only \(m/4\) additional ancilla. Notably, the output MSB of our parallelized variant retains its counting role and serves as the all-zero flag.

  \item \textbf{Polymorphic Architectural Design}: By leveraging the logical symmetry between zero and one detection, our design achieves functional polymorphism. This allows the core architecture to dynamically toggle between QLZC and QLOC configurations via a single mode-selection qubit. 

  \item \textbf{Resource Efficiency}: For $4$-qubit inputs (Table~\ref{tab:4bit_comparison}), our TA-P-OP $4$-QLZC outperforms Circuit 2 of~\cite{Q_LZD} across all evaluated metrics, reducing $\TGATE$-count, $\TGATE$-depth, ancilla count, width, and depth by 40.00\%, 60.00\%, 50.00\%, 20.00\%, and 47.54\%, respectively.  Our architecture remains shallower even without $\GATE{T\text{-}AND}$ integration, validating the algorithm's efficiency—the same trend holds for higher qubit counts (Table~\ref{tab:8bit_comparison}). Furthermore, we evaluate the sequential, parallel, and fan-out-control designs in Table~\ref{tab:tc-td-analysis-nop-op} to analyze the trade-offs of parallelization.
\end{enumerate}

For brevity, full resource derivations are deferred to~\cite{Supplementary}. The detailed cost analyses are presented as follows.

\begin{table*}[t]
\centering
\begin{tabular}{l c c c c c c}
\hline
\textbf{Circuit} & \textbf{T-count} & \textbf{T-depth}\\
\hline
TA-OP $m$-QLZC    & $4m-4 \;=\; \mathcal{O}\big(m)$ 
                    & $m \;=\; \mathcal{O}\big(m) $\\[2pt]
TA-OP $m$-PQLZC    & $\!\!\!\!3m\!+\!\log_2\! m(\tfrac{7}{2}\log_2\! m\!+\!\tfrac{1}{2})\!-\!\tfrac{23}{2} = \mathcal{O}\big(m)$ 
                    & $\log_2\! m(\tfrac{3}{2}\log_2\! m\!-\!\tfrac{1}{2})\!+\!\tfrac{1}{2} = \mathcal{O}\big((\log_2 m)^2\big)$\\[2pt]
FO-TA-OP $m$-PQLZC    & $\!\!\!\!3m\!+\!\log_2\! m(\tfrac{7}{2}\log_2\! m\!+\!\tfrac{1}{2})\!-\!\tfrac{23}{2} = \mathcal{O}\big(m)$ 
                    & $3\log_2 m-2=\mathcal{O}(\log_2 m)$\\[1pt]
\hline
\hline
\textbf{Circuit} & \textbf{\#Ancilla} & \textbf{\#Garbage Ancilla}\\
\hline
TA-OP $m$-QLZC    & $ m\;=\; \mathcal{O}\big(m) $ 
                    & $ 0\;=\; \mathcal{O}\big(1) $\\[2pt]
TA-OP $m$-PQLZC    & $ \tfrac{5}{4}m - \log_2 m - 2\;=\; \mathcal{O}\big(m) $ 
                    & $ m - \log_2 m - 2\;=\; \mathcal{O}\big(m) $\\[2pt]
FO-TA-OP $m$-PQLZC    & $ \tfrac{3}{2}m - \log_2 m - 2\;=\; \mathcal{O}\big(m) $ 
                    & $ m - \log_2 m - 2\;=\; \mathcal{O}\big(m) $\\[1pt]
\hline
\hline
\textbf{Circuit} & \textbf{Width} & \textbf{Depth}\\
\hline
TA-OP $m$-QLZC    & $2m + \lfloor \log_2 m \rfloor + 1 \;=\; \mathcal{O}\big(m) $ 
                    & $14(m - 1)\;+2m \;=\; \mathcal{O}\big(m) $ \\[2pt]
TA-OP $m$-PQLZC    & $\tfrac{9}{4}m - 1 \;=\; \mathcal{O}\big(m) $ 
                    & $\log_2\! m(\tfrac{9}{2}\log_2\! m\!+\!\tfrac{15}{2})\!-\!1  = \mathcal{O}\big((\log_2 m)^2\big)$\\[2pt]
FO-TA-OP $m$-PQLZC    & $\tfrac{5}{2}m - 1 \;=\; \mathcal{O}\big(m) $ 
                    & $12\log_2 m+8+2\lg((\log_2 m-1)!)=\mathcal{O}\!(\log_2 m\,\log_2\log_2 m)$\\[1pt]
\hline
\end{tabular}
\caption{Cost analysis of TA-OP $m$-QLZC, TA-OP $m$-PQLZC and FO-TA-OP $m$-PQLZC.}
\label{tab:tc-td-analysis-nop-op}
\end{table*}

\begin{table*}[tb]
\centering
\begin{tabular}{|c|c|r|c|c|c|c|c|}
\hline
\textbf{$n$} &
\textbf{$X$ (Dec)} &
\textbf{$X$ (Bin)$\hspace{1.75cm}$} &
\textbf{\makecell{$\MATHFN{LZC}(X)$ \\ (True)}} &
\textbf{\makecell{TA-OP \\ QLZC}} &
\textbf{\makecell{TA-OP \\ PQLZC}}&
\textbf{\makecell{FO-TA-OP \\ PQLZC}}&
\textbf{\makecell{reconfigurable QLZOC \\ ($c=1$)}} \\
\hline
11 & 0        & 00000000000                      & 11 & 11 & 11 & 11& 11 \\ \hline
13 & 1        & 0000000000001                    & 12 & 12 & 12 & 12& 12 \\ \hline
16 & 291      & 0000000100100011                 & 7  & 7  & 7  & 7 & 7  \\ \hline
20 & 241      & 00000000000011110001             & 12 & 12 & 12 & 12& 12 \\ \hline
24 & 42480    & 000000001010010111110000         & 8  & 8  & 8  & 8 & 8  \\ \hline
28 & 8388608  & 0000100000000000000000000000     & 4  & 4  & 4  & 4 & 4  \\ \hline
32 & 15790320 & 00000000111100001111000011110000 & 8  & 8  & 8  & 8 & 8  \\ \hline
\end{tabular}
\caption{Correctness check of TA-OP $m$-QLZC, TA-OP $m$-PQLZC, FO-TA-OP $m$-PQLZC and reconfigurable QLZOC under control $c=1$ (QLZC mode) against the ground-truth.}
\label{tab:lzd_correctness}
\end{table*}

\begin{table*}[tb]
\centering
\begin{tabular}{|c|c|r|c|c|c|c|c|}
\hline
\textbf{$n$} &
\textbf{$X$ (Dec)} &
\textbf{$X$ (Bin)$\hspace{1.75cm}$} &
\textbf{\makecell{$\MATHFN{LOC}(X)$ \\ (True)}} &
\textbf{\makecell{TA-OP \\ QLOC}} &
\textbf{\makecell{TA-OP \\ PQLOC}}&
\textbf{\makecell{FO-TA-OP \\ PQLOC}}&
\textbf{\makecell{reconfigurable QLZOC \\ ($c=0$)}} \\
\hline
11 & 2047       & 11111111111                      & 11 & 11 & 11 & 11 & 11 \\ \hline
13 & 8190       & 1111111111110                    & 12 & 12 & 12 & 12 & 12 \\ \hline
16 & 65475      & 1111111111000011                 & 10 & 10 & 10 & 10 & 10 \\ \hline
20 & 1044497    & 11111110111100010001             & 7  & 7  & 7  & 7  & 7  \\ \hline
24 & 16711680   & 111111110000000000000000         & 8  & 8  & 8  & 8  & 8  \\ \hline
28 & 260046848  & 1111011110000000000000000000     & 4  & 4  & 4  & 4  & 4  \\ \hline
32 & 4043309040 & 11110000111100001111000011110000 & 4  & 4  & 4  & 4  & 4  \\ \hline
\end{tabular}
\caption{Correctness check of TA-OP $m$-QLOC, TA-OP $m$-PQLOC, FO-TA-OP $m$-PQLOC and reconfigurable QLZOC under control $c=0$ (QLOC mode) against the ground-truth.}
\label{tab:loc_correctness}
\end{table*}


Since~\cite{Q_LZD} outperforms previous designs~\cite{Q_Fdivision_LZD, Q_Fadder_LZD,Q_Fadder_LZC}, we adopt it as our primary benchmark. 
Following Sec.~\ref{subsec:Temporary-logical-AND-gate}, we assign 4 $\TGATE$-count / 2 $\TGATE$-depth to each $\GATE{T\text{-}AND}$ and 7 $\TGATE$-count / 3 $\TGATE$-depth to each $\GATE{CCX}$ gate for all evaluations. 

First, Table~\ref{tab:4bit_comparison} compares our $4$-qubit input QLZC against two representative circuits from~\cite{Q_LZD}: ``Circuit~3'' (minimum-qubit) and ``Circuit~2'' (best-depth). 
Compared with Circuit~3, all three variants (P-OP $4$-QLZC, TA-P-OP $4$-QLZC, and TA-OP $4$-QLZC) improve $\TGATE$-count/$\TGATE$-depth and overall depth, achieving substantial depth reductions---up to \DTfoVtr\ in overall circuit depth and \TDTfoVtr\ in $\TGATE$-depth in the best cases. 
Against Circuit~2, TA-P-OP $4$-QLZC demonstrates a clear advantage across all evaluated metrics, reducing $\TGATE$-count, $\TGATE$-depth, ancilla count, width, and depth by 40.00\%, 60.00\%, 50.00\%, 20.00\%, and 47.54\%, respectively. 
Overall, TA-P-OP $4$-QLZC outperforms all evaluated metrics. Our architecture remains shallower regardless of $\GATE{T\text{-}AND}$ integration, validating the algorithm's efficiency—a trend mirrored in the $8$-qubit results (Table~\ref{tab:8bit_comparison}). This demonstrates that quantum-native constructions significantly surpass classical Boolean mappings, highlighting the necessity of specialized quantum intuition for efficient fault-tolerant circuits.

Next, we analyze the TA-OP $m$-QLZC, TA-P-OP $m$-PQLZC, and FO-TA-P-OP $m$-PQLZC designs in Table~\ref{tab:tc-td-analysis-nop-op} to evaluate the trade-offs of parallelization for a given resource budget. It can be observed that each circuit exhibits distinct advantages. 
The TA-OP $m$-QLZC requires fewer ancilla and produces no garbage, making it preferable when qubit resources are strictly limited. 
In contrast, the TA-P-OP $m$-PQLZC reduces the $\TGATE$-count and achieves lower $\TGATE$-depth and overall depth, providing a more favorable depth--non-Clifford cost trade-off. 
Furthermore, the FO-TA-P-OP $m$-PQLZC achieves a superior $\TGATE$-depth reduction to $\mathcal{O}(\log_2 m)$ and an overall depth of $\mathcal{O}(\log_2 m \log_2 \log_2 m)$, at the expense of only $m/4$ additional reusable ancilla. 
Designers may thus choose between the three implementations depending on the specific hardware or application requirements.




\section{Experimental Results} \label{sec:ExperimentalResults}
This subsection presents the experimental results for the various quantum circuits proposed in this work. This includes TA-OP $m$-QLZC, TA-OP $m$-PQLZC, and FO-TA-OP $m$-PQLZC. Furthermore, we evaluate their corresponding QLOC versions and the reconfigurable QLZOC architecture, which supports both leading-zero and leading-one counting modes. 

We employ \texttt{Qiskit}~\cite{qiskit} for small-to-medium instances and switch to \texttt{SliQSim}~\cite{sliqsim} for larger qubit counts, where direct simulation becomes computationally expensive.

To verify the functional integrity of these designs, each circuit is tested against ground-truth values across various register widths $n$, with results summarized in the following tables. 
While the tables list results up to 32 bits for brevity, higher bit-width tests were also performed and verified against ground-truth values. The implementation can be found in~\cite{Supplementary}.

As shown in Table~\ref{tab:lzd_correctness}, for representative test vectors $X$, all circuits output match the ground-truth reference $\MATHFN{LZC}(X)$, confirming that the proposed QLZC behave as intended. 
We note that $m$-PQLZC natively supports widths $m = 8\cdot 2^p$
. For arbitrary widths, we pad the input by appending zeros to the least-significant end to obtain a supported size without altering $\MATHFN{LZC}(X)$ or correctness.

As shown in Table~\ref{tab:lzd_correctness}, for various representative test vectors $X$, all proposed circuits of QLZC yield outputs that match the ground-truth $\MATHFN{LZC}(X)$, confirming their functional correctness. 
Since $m$-PQLZC natively supports register widths $m = 8 \cdot 2^p$
, inputs of arbitrary width are padded with zeros at the least-significant end; this ensures compatibility without altering the $\MATHFN{LZC}(X)$ value or the circuit's integrity. 
Similar verification results for the QLOC variants and the reconfigurable QLZOC architecture are summarized in Table~\ref{tab:loc_correctness} and Table~\ref{tab:lzd_correctness}, respectively, all demonstrating full alignment with their corresponding ground-truth references.

These results serve to substantiate our theoretical analyses and to empirically validate the correctness and expected performance of circuit.
\section{Conclusions} \label{sec:Conclusions}
This paper introduces a resource-efficient, scalable QLZOC architecture, serving as a foundational building block for quantum normalization and numerical scaling. By reformulating the counting logic into conditional gate operations, we provide modular designs that prioritize $\TGATE$-efficiency. Our FO-PQLZOC achieves a substantial $\TGATE$-depth reduction to $\mathcal{O}(\log_2 m)$ with minimal ancilla overhead, while the polymorphism of our design enables dynamic reconfiguration between LZC and LOC modes with negligible Clifford overhead. Experimental simulations across various register widths $n$ empirically validate the functional integrity and correctness of our designs. These improvements—highlighted by 60.00\% $\TGATE$-depth and 40.00\% $\TGATE$-count reductions in the 4-qubit case—establish our QLZOC as a robust and flexible component for next-generation quantum arithmetic and large-scale processors.


Future research will focus on integrating this QLZOC architecture into full-scale quantum floating-point units and logarithmic arithmetic modules, while investigating further $\TGATE$-depth optimizations using advanced gate-cancellation techniques. Additionally, we aim to explore the automated synthesis and integration of these modular blocks into comprehensive quantum arithmetic logic units (ALU) for large-scale applications.

\clearpage

\nocite{*}
\bibliographystyle{IEEEtran}
\bibliography{QLZC_reference}


\end{document}